\documentclass[twocolumn,notitlepage,aps,prb,10pt,superscriptaddress,floatfix,letterpaper
]{revtex4-1}

\usepackage{hyperref}
\usepackage[usenames,dvipsnames]{color}
\usepackage{amsmath}
\usepackage{amssymb}
\usepackage{amsthm}
\usepackage{graphicx}
\usepackage{epstopdf}
\usepackage{todonotes}
\usepackage[normalem]{ulem}

\usepackage{tikz,bm}
\usetikzlibrary{positioning,arrows}

\usetikzlibrary{arrows,shapes.geometric,backgrounds,positioning,calc}
\newtheorem{theorem}{Theorem}
\newtheorem{corollary}{Corollary}[theorem]

\begin{document}

\title{Chebyshev polynomial representation of imaginary time response functions}
\author{Emanuel Gull}%
\affiliation{%
 Department of Physics, University of Michigan, Ann Arbor, Michigan 48109, USA
}%
\author{Sergei Iskakov}
\affiliation{%
 Department of Physics, University of Michigan, Ann Arbor, Michigan 48109, USA
}%
\author{Igor Krivenko}
\affiliation{%
 Department of Physics, University of Michigan, Ann Arbor, Michigan 48109, USA
}%
\author{Alexander A. Rusakov}
\affiliation{%
 Department of Chemistry, University of Michigan, Ann Arbor, Michigan 48109, USA
}%
\author{Dominika Zgid}%
\affiliation{%
 Department of Chemistry, University of Michigan, Ann Arbor, Michigan 48109, USA
}%
\date{\today}
\newcommand{\sump}{\ensuremath{\sideset{}{'}\sum}} % primed sum
\newcommand{\sumplim}[2]{ \sideset{}{'}\sum_{#1}^{#2}}
\newcommand{\sumponelim}[1]{ \sideset{}{'}\sum_{#1}}
\newcommand\myatop[2]{\genfrac{}{}{0pt}{}{#1}{#2}}

\begin{abstract}
Problems of finite temperature quantum statistical mechanics can be formulated in terms of imaginary (Euclidean) time Green's functions and
self-energies. In the context of realistic Hamiltonians, the large energy scale  of the Hamiltonian (as compared to temperature) necessitates a
very precise representation of these functions. In this paper, we explore the representation of Green's functions and self-energies in terms of series
of Chebyshev polynomials. We show that many operations, including convolutions, Fourier transforms, and the solution of the Dyson equation,
can straightforwardly be expressed in terms of the series expansion coefficients. We then compare the accuracy of the Chebyshev
representation for realistic systems with the uniform-power grid representation, which is most commonly used in this context.
\end{abstract}

\maketitle

\section{Introduction}
The equilibrium properties of interacting quantum systems at finite temperature can be described by the Matsubara formalism of quantum
statistical mechanics.\cite{AGD75} In this formalism, single- and two-particle quantities are expressed in terms of Green's functions,
self-energies, susceptibilities, and vertex functions in imaginary time.

The imaginary time formalism has a long tradition in the calculation of properties of interacting systems,\cite{Bloch58,Luttinger60} and
weak coupling methods such as the random phase approximation,\cite{Bohm53,GellMann57} the self-consistent second order
approximation,\cite{Dahlen05,Phillips14,Phillips15,doi:10.1021/acs.jctc.5b00884,Kananenka16,Rusakov16,Welden16} or the $GW$ method\cite{Hedin65} can be formulated in terms of imaginary time Green's
functions and self-energies. 
Numerical algorithms, including lattice Monte Carlo methods,\cite{Blankenbecler81} impurity solver
algorithms,\cite{Hirsch86,Rubtsov05,Werner06,Gull08_ctaux,Gull11_RMP} diagrammatic Monte Carlo methods\cite{Prokofev07} are similarly based
on finite-temperature Green's function formalism, as are some implementations of the dynamical mean field theory\cite{Georges96} and its
extensions.\cite{Toschi07,Rubtsov08,Maier05}

Finite temperature fermionic (bosonic) imaginary time Green's functions are antiperiodic (periodic) functions with period $\beta$ and
can be reduced to the interval $[0,\beta]$.
In most of the applications mentioned above, they are sampled on a uniform grid, typically with  $10^2$ to $10^4$ grid
discretization points. However, a uniform representation of the function is only efficient in effective model systems. In large
multi-orbital systems, and especially in systems with realistic band structures, an accurate representation of Green's functions with a
uniform discretization would require millions to billions of time slices per Green's function element, as the wide energy spacing of realistic
Hamiltonians results in features on very small time scales. Therefore, more compact representations of Green's functions are needed in this context.

A first attempt at constructing a more compact representation, the `uniform power mesh', was proposed by \textcite{Ku00}, see also
Ref.~\onlinecite{Ku02}. There, a set of logarithmically spaced nodes is chosen on the imaginary time interval. The Green's function is
then uniformly discretized between those nodes, using a constant number of points for each interval. This leads to a clustering of points near
$0$ and $\beta$, where much of the rapid change of Green's functions for low-lying excitations takes place. Later, Legendre
polynomial representations\cite{Boehnke11} were pioneered in the context of continuous-time Monte Carlo methods, where the compactness of the
representation could reduce the number of observables that needed to be accounted for, and in the context of analytical
continuation,\cite{Shinaoka17AC} where an intermediate basis of a singular value decomposed analytic continuation
kernel\cite{Shinaoka17Basis} could further reduce the number of coefficients. This was followed by progress in the context of perturbation methods for realistic
systems,\cite{doi:10.1021/acs.jctc.5b00884} where the combination of uniform power meshes and Legendre polynomial expansions drastically
reduced the size of the imaginary time grid. In Matsubara frequency space, Ref.~\onlinecite{Kananenka16}
showed that much of the Matsubara frequency dependence of Green's functions and self-energies can be represented by interpolation functions,
thereby vastly reducing the number of frequencies required to obtain accurate results.

For practical use in real materials simulations, a set of basis functions for imaginary time Green's functions and self-energies should
satisfy at least the following criteria. First and foremost, it should be possible to represent the large energy spread of typical
interacting systems with a small number of coefficients. Second, it should be straightforward to confirm that the representation is fully
converged, i.e. that basis truncation errors are small. Finally, the mathematics of performing typical operations on Green's functions, such
as evaluating a self-energy, a polarization bubble, a Dyson equation, or Fourier transforming data to frequency space and evaluating
energies should be straightforward, both analytically and in terms of the numerical effort.

The representations mentioned above satisfy some but not all of these requirements. In this paper, we therefore introduce an alternative representation of
imaginary time Green's functions, based on approximating the Green's functions by a sum of scaled Chebyshev polynomials of the first kind. 
We test the performance of this expansion explicitly for a variety of systems in realistic basis sets, including periodic solids. We examine how 
the number of coefficients converges as a function of temperature, basis set size, and precision required, and we show how Fourier transforms
and Dyson equations can be solved directly in Chebyshev space.

The remainder of this paper is organized as follows. In section \ref{sec:method}, we present the detailed derivation of the method. In
section \ref{sec:results}, we list and discuss the numerical results of our method as applied to realistic molecular and solid state
calculations. We present conclusions in section \ref{sec:conclusions}.

\section{Method}\label{sec:method}

\subsection{Chebyshev expansion of response functions}
Imaginary time Green's functions $G^{\nu\mu}(\tau)=-\langle c_\nu(\tau)c_\mu^\dagger(0)\rangle$ are defined for $0\leq \tau\leq\beta$, where
$\beta$ is the inverse temperature and Greek letters correspond to orbital indices. Outside this interval, fermionic Green's functions
satisfy $\beta$ anti-periodicity $G(-\tau)=-G(\beta-\tau)$, whereas bosonic response functions are $\beta$-periodic, $\chi(-\tau)=\chi(\beta
- \tau)$. In the following we will work on the interval $[0,\beta]$, and use the mapping $x(\tau)=\frac{2\tau}{\beta}-1$ to map it to the
interval $[-1,1]$.

The Chebyshev polynomials of the first kind, $T_{j}(x)$,\cite{Mason03} form a complete basis for bounded functions in this interval. Any Green's
function, or other response function can therefore be expanded into a sum of Chebyshev polynomials and approximated by a truncated Chebyshev series
\begin{align}
G^{\nu\mu}(x)&\approx\sumplim{j=0}{m} g_j^{\nu\mu}T_j(x)=\sum_{j=0}^{m}
g_j^{\nu\mu}T_j(x)-\frac{1}{2}g_0^{\nu\mu}\label{Eq:chebyshev_series},\\
g_j^{\nu\mu}&=\frac{2}{\pi}\int_{-1}^1\frac{G^{\nu\mu}(x)T_j(x)}{\sqrt{1-x^2}}dx.\label{Eq:gdef}
\end{align}
The primed sum denotes the special treatment of the coefficient $g_0$ customary in this context.~\cite{NR93} Based on the discrete
orthogonality properties of the Chebyshev polynomials,~\cite{NR93} if the values of $G^{\nu\mu}(x)$ are known at the zeros of the $m$th
Chebyshev polynomial, $x_{k}=\cos\left(\frac{2k-1}{2m}\pi\right),\ \ k=1,\ldots,m$, the calculation of the coefficients in Eq.~\ref{Eq:gdef}
simplifies to a discrete cosine transform. In addition, values of the Chebyshev approximant anywhere in the interval $0\leq\tau\leq\beta$
can be obtained from $g_j^{\nu\mu}$ using Clenshaw recursion relations.\cite{Clenshaw55}

Chebyshev representations are particularly efficient for approximating analytic functions on the interval $[-1,1]$, as approximation theory
guarantees that the magnitude of the coefficients $g_j^{\nu\mu}$ decays at least exponentially as $j\rightarrow\infty$, and that the maximum
difference between $G$ and its Chebyshev approximant decreases exponentially.\cite{Mason03} The fermionic and bosonic imaginary time Green's
functions, polarization functions, self-energies, and response functions appearing in finite-temperature many-body theory are all analytic
functions between $0$ and $\beta$.

As we will show in Sec.~\ref{sec:results}, fast convergence of Green's functions and self-energies  with the number of Chebyshev
coefficients is observed, and the discrete cosine transforms and recursion relations allow for quick numerical operations on the data in
practice. We find that our examples converge to a precision of $10^{-10}$ within about $40$ coefficients for the simple realistic systems, such as
hydrogen molecules, whereas around $500$ Chebyshev nodes are required to describe a krypton atom in a pseudopotential approximation.

\subsection{Convolutions}
Products of Matsubara Green's functions correspond to convolutions in imaginary time.
The convolution
\begin{align}
A(t)=\int_0^\beta d\tau B(t-\tau)C(\tau)\label{Eq:conv}
\end{align}
with $A$, $B$, and $C$ Green's functions or self-energies, requires careful treatment of the discontinuity of $B(t-\tau)$ at $t=\tau$,
so that standard Chebyshev convolution formulas \cite{Hale14} cannot be applied. Instead, we express Eq.~\ref{Eq:conv} by expanding the rescaled
integral into Chebyshev components (appropriately rescaling the zero coefficients)
\begin{align}
\sump_ja_j^{\nu\mu} T_j(x)&=\sum_{kl\xi}b_k^{\nu\xi} c_l^{\xi\mu} I_{kl}(x),\\
I_{kl}(x)&=\frac{\beta}{2}\Big[\int_{-1}^x T_k(x-y-1)T_l(y)dy\\&\mp\int_x^1 T_k(x-y+1)T_l(y) dy\Big]\nonumber,
\end{align}
where the minus (plus) sign corresponds to a convolution of fermionic (bosonic) functions, $I_{kl}(x)$ is an integral of polynomials in
$[-1,1]$ and can therefore be written as a Chebyshev series,
\begin{align}
I_{kl}(x)=\frac{\beta}{2}\sump_j t^j_{kl} T_j(x),
\end{align}
resulting in the formulation of the fermionic convolution as a matrix multiplication
\begin{align}
a_j^{\nu\mu}=\frac{\beta}{2}\sum_{kl\xi}b_k^{\nu\xi}c_l^{\xi\mu}t^j_{kl}.
\end{align}
This representation becomes more efficient than the Fourier representation whenever a very large number of Fourier components is
required. A detailed derivation of recursion relations for bosonic and fermionic integrals $t^j_{kl}$ is provided in the appendix.

\subsection{Dyson Equation}\label{sec:Dyson}
Most diagrammatic algorithms are formulated in imaginary time, where the  interaction vertex $V_{pqrs}$ is instantaneous.
However, most contain a step for solving a Dyson equation, either for adjusting a chemical potential to the desired particle number or to
obtain self-consistent propagators. This Dyson equation $G=G_0+G_0\Sigma G$ is most conveniently expressed in frequency space, where it can
be solved for each frequency independently.
In imaginary time, the Dyson equation determining $G$, given $G_0$ and $\Sigma$, corresponds to a Fredholm integral equation of the second
kind.~\cite{Zemyan12,Sakkinen15} As in the case of the Fourier transforms and convolutions, the discontinuity at zero and the highly
non-uniform nature of the Green's functions make uniform discretizations \cite{Dahlen05} inefficient. Defining $B(t)=\int d\tau
G_0(t-\tau)\Sigma(\tau)$ and expanding into Chebyshev coefficients, we obtain the equation
\begin{align}
g_j^{\nu\mu}=g_{(0)j}^{\nu\mu}+\frac{\beta}{2}\sum_{kl\xi}b_k^{\nu\xi}g_l^{\xi\mu}t^j_{kl}
\end{align}
with $t^j_{kl}$ defined as above. This linear system can be
recast as $\sum_{j\xi} A_{ij}^{\nu\xi}g_j^{\xi\mu}=g_{0i}^{\nu\mu}$ with a matrix $A_{ij}^{\nu\mu}=\delta_{ij}\delta_{\nu\mu}-\frac{\beta}{2}\sum_k b_k^{\nu\mu}t^j_{kl}$.
The solution of the Fredholm integral equation is thereby mapped onto the solution of a system of linear equations, 
bypassing the Matsubara domain entirely.

\subsection{Fourier Transforms}\label{sec:Fourier}
Fourier transforms between time and frequency domains require a careful treatment of the Green's function around $\tau=0$. At this point,
fermionic Green's functions are discontinuous due to the fermion anticommutation relation. This discontinuity is usually absorbed in an
explicit treatment of the short time (high frequency) behavior using high frequency expansions and suitable model
functions.\cite{Comanac07,Rusakov14,Gull11_RMP} Even with this high-frequency treatment, the number of Matsubara frequencies required for
accurate energies and spectra of realistic systems at low temperature is very large. This is because the spacing of the Matsubara points is
given by the inverse temperature, whereas the location of the main features of the function is given by the energy scale of the Hamiltonian.
In realistic systems, these energy scales are different by many orders of magnitude, requiring millions to billions of frequency points.
Adaptive grids, such as the one developed in Ref.~\onlinecite{Kananenka16}, provide a partial solution to this problem.

Fourier transforms of Chebyshev polynomials to Matsubara frequencies 
$\omega_n=\frac{(2n+\zeta)\pi}{\beta}$ ($\zeta=0$ for bosons and $\zeta=1$ for 
fermions), are obtained by evaluating the integral\cite{Fokas12}
\begin{align}\label{Eq:ftdef}
\mathcal{F}(T_{m})(i\omega_n) &=\int_{0}^{\beta}d\tau T_{m}(x(\tau))e^{i \omega_n \tau} = \nonumber \\
& = \frac{\beta}{2}\int_{-1}^{1}dx T_{m}(x)e^{i\lambda_n \frac{{x+1}}{2}}
= F^\zeta_{mn}
\end{align}
for dimensionless Matsubara frequencies $\lambda_n=(2n+\zeta)\pi$. In the special case of bosonic Matsubara frequency zero, we find that
\begin{equation}
F^0_{m0} = \frac{\beta}{2}\int_{-1}^1dx T_m(x) =
\frac{\beta}{2}
\begin{cases}
\frac{1+(-1)^m}{1 - m^2}, &m\neq1,\\
0, &m=1.
\end{cases}
\end{equation}
For all non-zero $\lambda_n$, partial integration yields
\begin{align}\label{eq:Imn}
I^\zeta_{m}(n)&=\int_{-1}^{1}dxT_{m}(x)e^{i\lambda_n\frac{x+1}{2}} =\left.\frac{2}{i\lambda_n}e^{i\lambda_n\frac{x+1}{2}}T_{m}(x)
\right|_{-1}^{1} \nonumber -\\ &-\frac{2}{i\lambda_n}\int_{-1}^{1}dxT_{m}'(x)e^{i\lambda_n\frac{x+1}{2}},
\end{align}
where the boundary term evaluates to
\begin{align}
\left.\frac{2}{i\lambda_n}e^{i\lambda_n\frac{x+1}{2}}T_{m}(x)\right|_{-1}^{1}=
2i\frac{(-1)^{m} - (-1)^\zeta}{\lambda_n}.
\end{align}
Using $T_{m}'(x)=mU_{m-1}(x),$ where $U_m(x)$ are
Chebyshev polynomials of the second kind related to $T_m(x)$ via
\begin{align}
U_{m}(x)=\begin{cases}
2\left(\sum_{j\ \text{odd}}^{m}T_{j}(x)\right) & m\ \text{odd},\\
2\left(\sum_{j\ \text{even}}^{m}T_{j}(x)\right)-1 & m\ \text{even},
\end{cases}
\end{align}
we transform the second integral in (\ref{eq:Imn}) as
\begin{align}
&\frac{2}{i\lambda_n}\int_{-1}^{1}dxT_{m}'(x)e^{i\lambda_n\frac{x+1}{2}}=
\frac{2m}{i\lambda_n}\int_{-1}^{1}U_{m-1}(x)e^{i\lambda_n
\frac{x+1}{2}}=\nonumber \\
&=\frac{2m}{i\lambda_n}\int_{-1}^{1}dx e^{i\lambda_n \frac{x+1}{2}}
\begin{cases}
2\sum_{j,\text{odd}}^{m-1} T_j(x), &m\ \text{even},\\
2\sum_{j,\text{even}}^{m-1} T_j(x) - 1, &m\ \text{odd}.
\end{cases}
\end{align}
This results in a recursion relation with respect to index $m$,
\begin{align}
&I^\zeta_m(n) = 2i\frac{(-1)^{m}-(-1)^\zeta}{\lambda_n} +\nonumber\\+&\frac{2im}{\lambda_n}
\begin{cases}
2\left(\sum_{j\ \text{odd}}^{m-1}I^\zeta_{j}(n)\right), & m\ \text{even},\\
2\left(\sum_{j\ \text{even}}^{m-1}I^\zeta_{j}(n)\right)
-I^\zeta_0(n), & m\ \text{odd},
\end{cases}
\end{align}
which we start by explicitly computing
$I^0_0(n) = 2\delta_{0,n}$ or $I^1_0(n) = 4i/\lambda_n$.
This recursion relation is unstable\cite{doi:10.1093/imanum/drq036} and therefore has to be
implemented in high precision arithmetic. With Eq.~\ref{Eq:ftdef} we then
write the Fourier transform as
\begin{align}\label{eq:FTMatMul}
\mathcal{F}(G)(i\omega_{n})=\sum_{j}g_{j}F^\zeta_{jn},
\end{align}
where the Fourier matrix $F^\zeta_{jn}$ is computed once and tabulated. The inverse transform is done by evaluating the function at the Chebyshev
nodes and using a discrete cosine transform to obtain the corresponding coefficients.

Accurate Fourier transforms and energy evaluations in Fourier space require high frequency expansion coefficients for the Green's function
to at least third order, so that $G(i\omega_n)=\frac{c_1}{i\omega_n}+\frac{c_2}{(i\omega_n)^2}+\frac{c_3}{(i\omega_n)^3}+O((i\omega_n)^{-4})$.
Fourier transform of the Green's function implies that $c_1= - (G(0^+) + G(\beta^-))$; $c_2=G'(0^+)+G'(\beta^-)$; and in general
$c_{k+1}=(-1)^{k+1}(G^{(k)}(0^+)+G^{(k)}(\beta^-))$.
These expansion coefficients are available to any order due to the identities
\begin{align}
\left.\frac{d^{p}T_{n}}{dx^{p}}\right|_{x=\pm1}=(\pm1)^{n+p}\prod_{k=0}^{p-1}\frac{n^{2}-k^{2}}{2k+1},
\end{align}
and in particular
\begin{align}
\left.\frac{dT_{n}}{dx}\right|_{x=\pm1}&=(\pm1)^{n}n^{2},\\
\left.\frac{d^{2}T_{n}}{dx^{2}}\right|_{x=\pm1}&=(\pm1)^{n}\frac{n^{4}-n^{2}}{3}.
\end{align}

\begin{figure}[bth]
\includegraphics[width=\columnwidth]{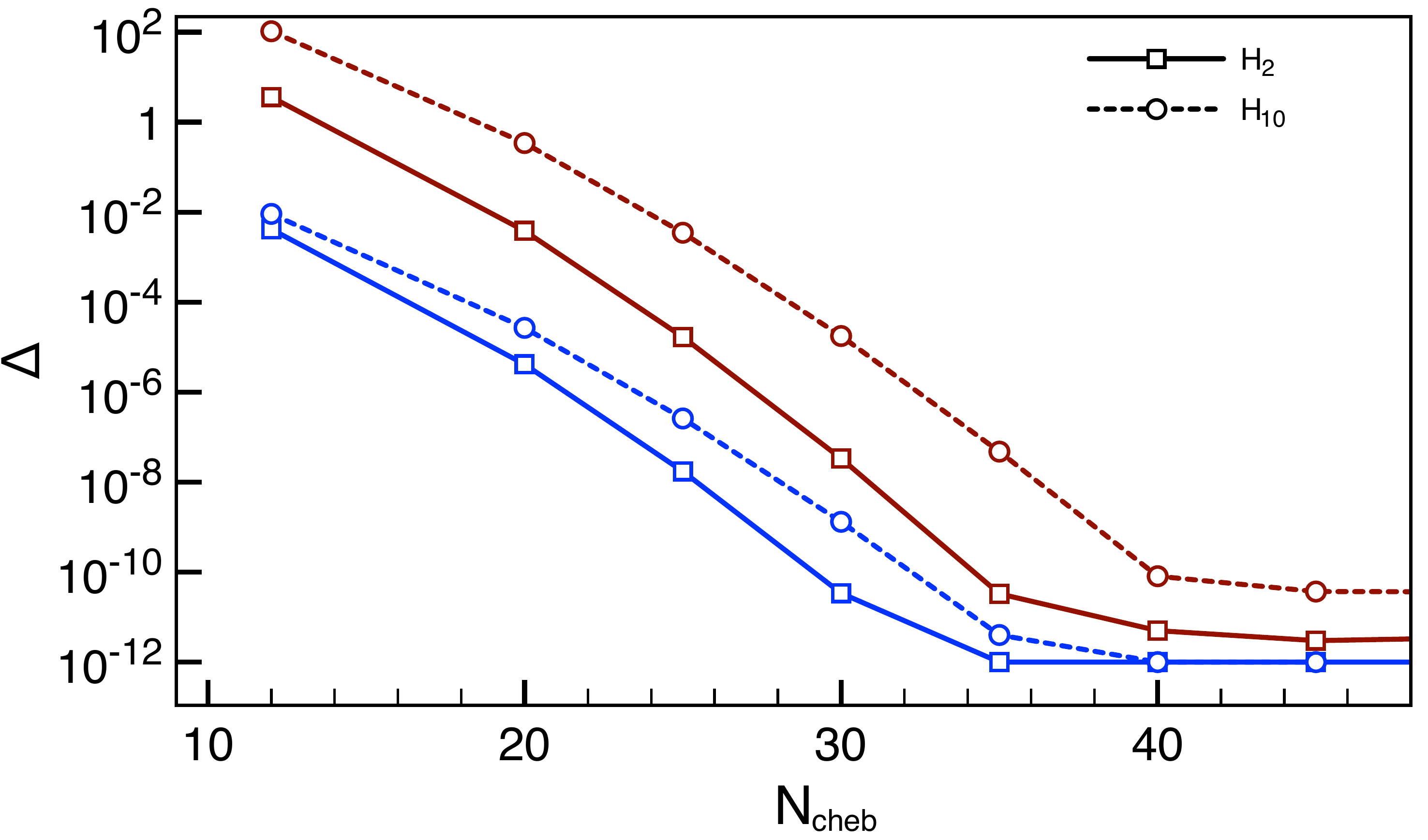}
\includegraphics[width=\columnwidth]{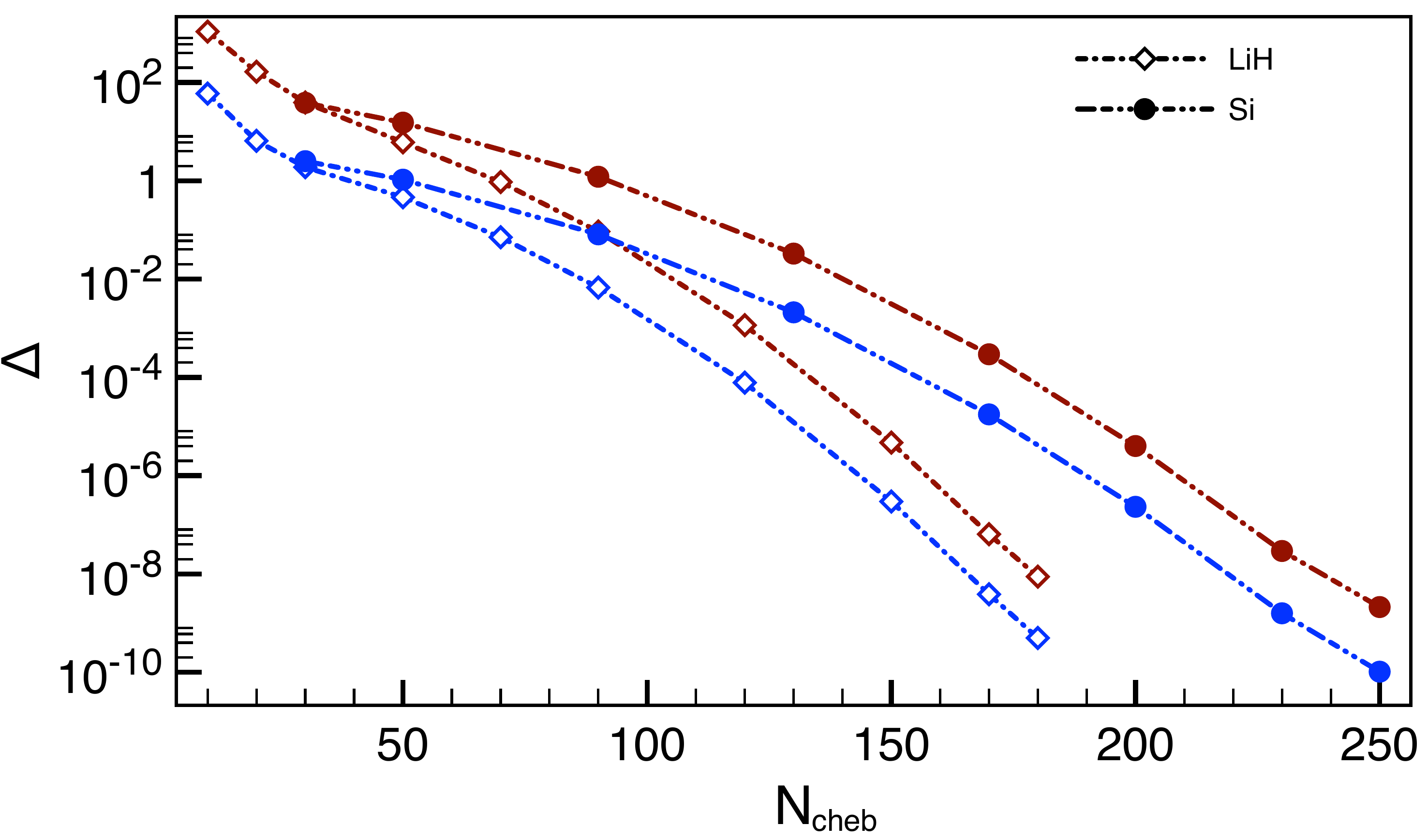}
\caption{Convergence of the Hartree-Fock Green's function with the number of Chebyshev polynomials. Red curves correspond to the sum 
of all differences with respect to the exact result. Blue curves correspond to the maximum difference when compared to the exact result. 
Top panel: H$_2$ molecule (open square) and H$_{10}$ molecule (open circle). Bottom panel: periodic one-dimensional LiH (open
diamond) and three-dimensional Si crystal (filled circle). Parameters as specified in Table~\ref{tab:Parameters}.}
\label{fig:ChebDiff}
\end{figure}

\section{Results}\label{sec:results}
To demonstrate the efficiency of the proposed method we consider four test systems. In order to make the results reproducible, we use
electronic structure systems in standardized basis sets that are well documented\cite{szabo1996modern} and readily
available.\cite{FellerBSE,SchuchardtBSE} The first two systems are hydrogen molecules, both as H$_2$ and as a one-dimensional chain of ten
hydrogen atoms. We use the minimal STO-3g basis\cite{Hehre69} and place the atoms at an inter-atomic distance of $d=1.5$\AA{}. These cases
are chosen as `easy' examples of realistic calculations. We also consider two periodic systems. First, a one-dimensional periodic LiH solid
in the triple-zeta quality basis set (pob-TZVP) from Ref.~\onlinecite{doi:10.1002/jcc.23153} and, 
second, a three-dimensional Si crystal in the following basis set: the innermost 1s, 2s, and 2p shells are replaced with the LANL2DZ effective core potentials,\cite{doi:10.1063/1.448799,doi:10.1063/1.448800,doi:10.1063/1.448975} while the basis functions for the outer 3s, 3p, and 3d shells are taken from the Si\_88-31G*\_nada\_1996 basis.\cite{doi:10.1002/SICI,doi:10.1021/jp002353q,Prencipe2006}
All systems were evaluated at an inverse temperature
of $\beta=100$Ha$^{-1}$. The detailed parameters are shown in Table~\ref{tab:Parameters}.

\begingroup
\squeezetable
\begin{table}[tbh]
\begin{ruledtabular}
\begin{center}
\begin{tabular}{ l | c | c | c }
\multicolumn{4}{c}{Molecular systems}\\
\hline
& Basis & \multicolumn{2}{c}{inter-atomic distance, \AA{}} \\
\hline
H$_{2}$& STO-3g& \multicolumn{2}{c}{1.5000} \\
\hline
H$_{10}$& STO-3g& \multicolumn{2}{c}{1.5000} \\
\hline\hline
\multicolumn{4}{c}{Periodic systems}\\
\hline
& Basis & Unit cell coordinates, \AA{} & Translation vectors, \AA{} \\
\hline
LiH& pob-TZVP& \begin{tabular}{@{}c@{}}Li 0.0 0.0 0.0\\H 1.671286 0.0 0.0 \end{tabular} & (3.342572, 0.0, 0.0) \\
\hline
Si& Custom, see text& \begin{tabular}{@{}c@{}}Si  0.0000    0.0000    0.0000\\Si  1.3575   1.3575    1.3575\end{tabular} & 
\begin{tabular}{@{}c@{}}
(0.0, 2.7150, 2.7150)\\
(2.7150, 0.0, 2.7150)\\
(2.7150, 2.7150, 0.0)
\end{tabular}\\
\end{tabular}
\end{center}
\end{ruledtabular}
\caption{\label{tab:Parameters}Geometries and basis sets for the systems used. 
All systems were evaluated at an inverse temperature of $\beta=100$Ha$^{-1}$.}
\end{table}
\endgroup

The exponential convergence predicted by theory can be observed in practice. Fig.~\ref{fig:ChebDiff} shows the convergence of the Chebyshev
Green's function to the exactly evaluated Hartree-Fock solution as a function of the number of coefficients. Shown is the difference $\Delta$ between the interpolated
Green's function and a reference Green's function evaluated analytically on a uniform-power grid (with 15 power points and 18 uniform
points between each pair of power points) as a function of the number of coefficients, both as the maximum deviation (dark red curves) and
as the sum of all deviations at all points (blue curves).  

It is evident that the Chebyshev approximation converges to the exactly evaluated Hartree-Fock solution as a function of the number of expansion coefficients, until numerical
roundoff errors are reached at a precision of $10^{-12}$. For the hydrogen systems used here, between $30$ and $40$ coefficients lead to a
maximum uncertainty of around $10^{-10}$. More complex systems require more coefficients, as illustrated by the one-dimensional LiH and
three-dimensional Si, which only reach a precision of $10^{-8}$ within $180$ and $230$ coefficients respectively (lower panel of
Fig.~\ref{fig:ChebDiff}). This convergence stems from the fast decay of the Chebyshev expansion coefficients. 
Fig.~\ref{fig:ChebCoeff} illustrates this point by showing the maximum magnitude of the Chebyshev coefficients of a given polynomial order
as a function of order. There, the noise floor is reached for around $40$ (H$_{2}$) or $80$ (H$_{10}$) coefficients, at a coefficient size
of around $10^{-16}$.

\begin{figure}[bth]
\includegraphics[width=\columnwidth]{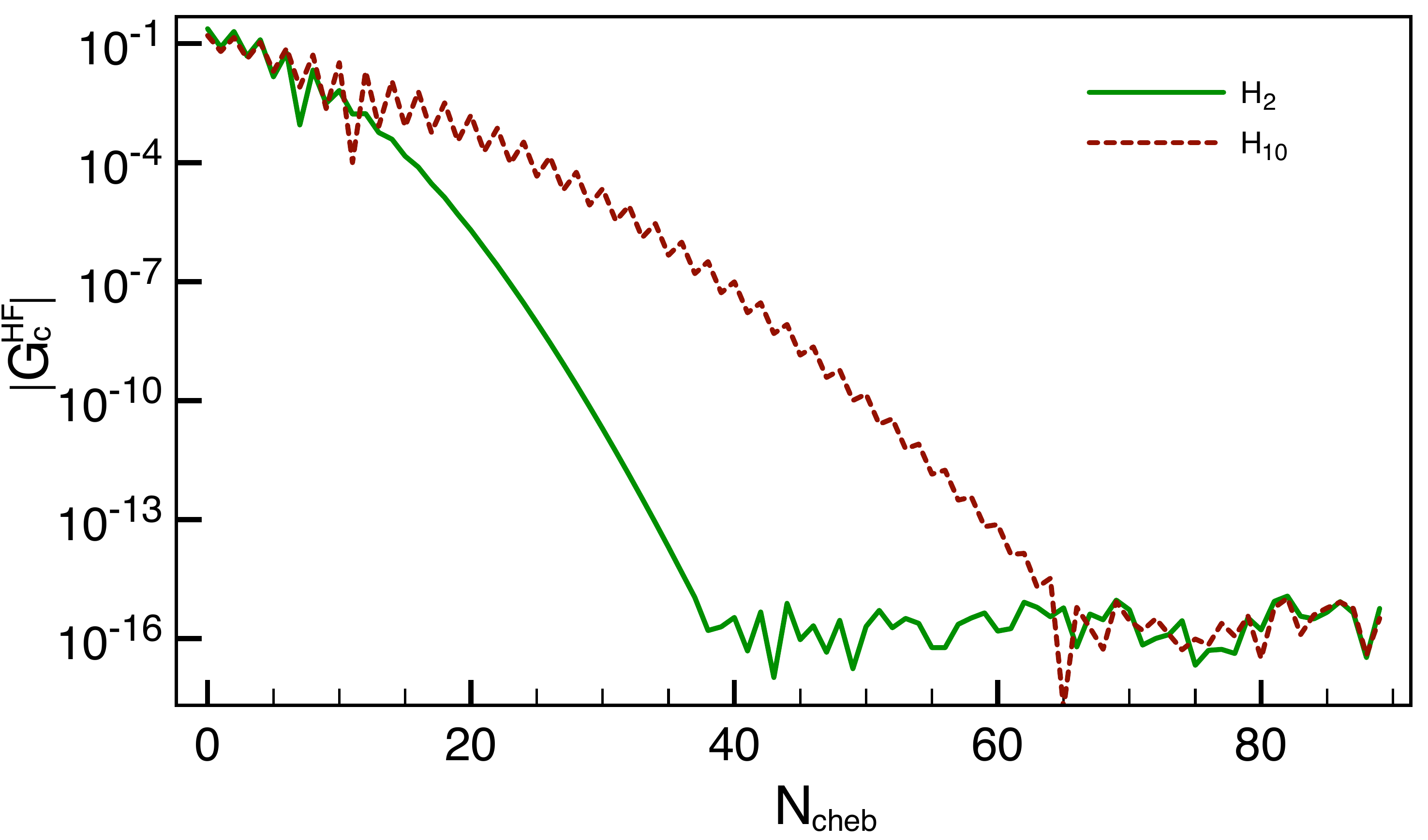}
\caption{Exponential decay of the Chebyshev coefficients for the  Hartree-Fock Green's function. H$_2$ molecule (green) and H$_{10}$
molecule (red, dashed) with parameters as specified in Table~\ref{tab:Parameters}.
}\label{fig:ChebCoeff}
\end{figure}

\begin{figure}[bth]
\includegraphics[width=\columnwidth]{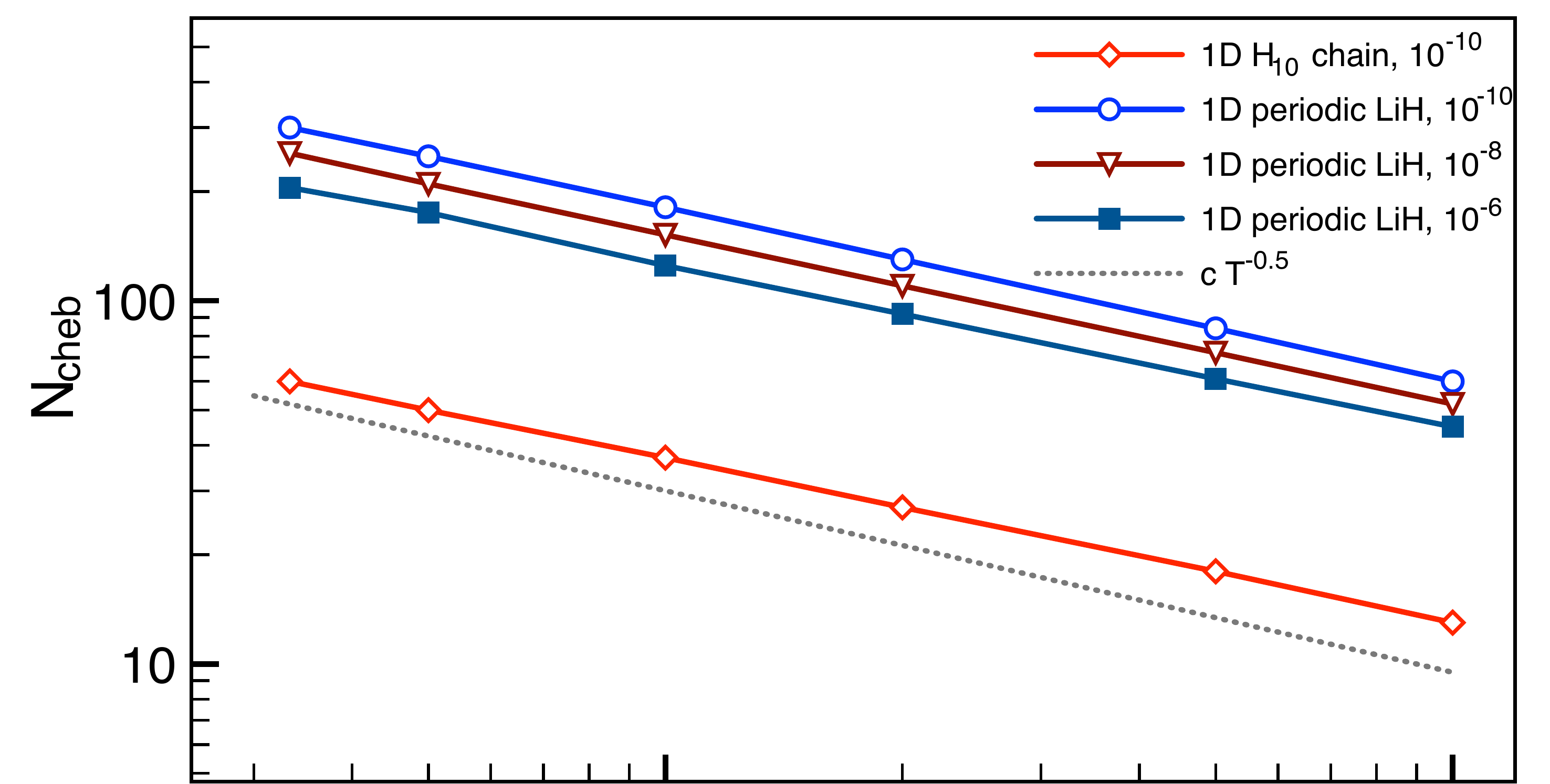}
\includegraphics[width=\columnwidth]{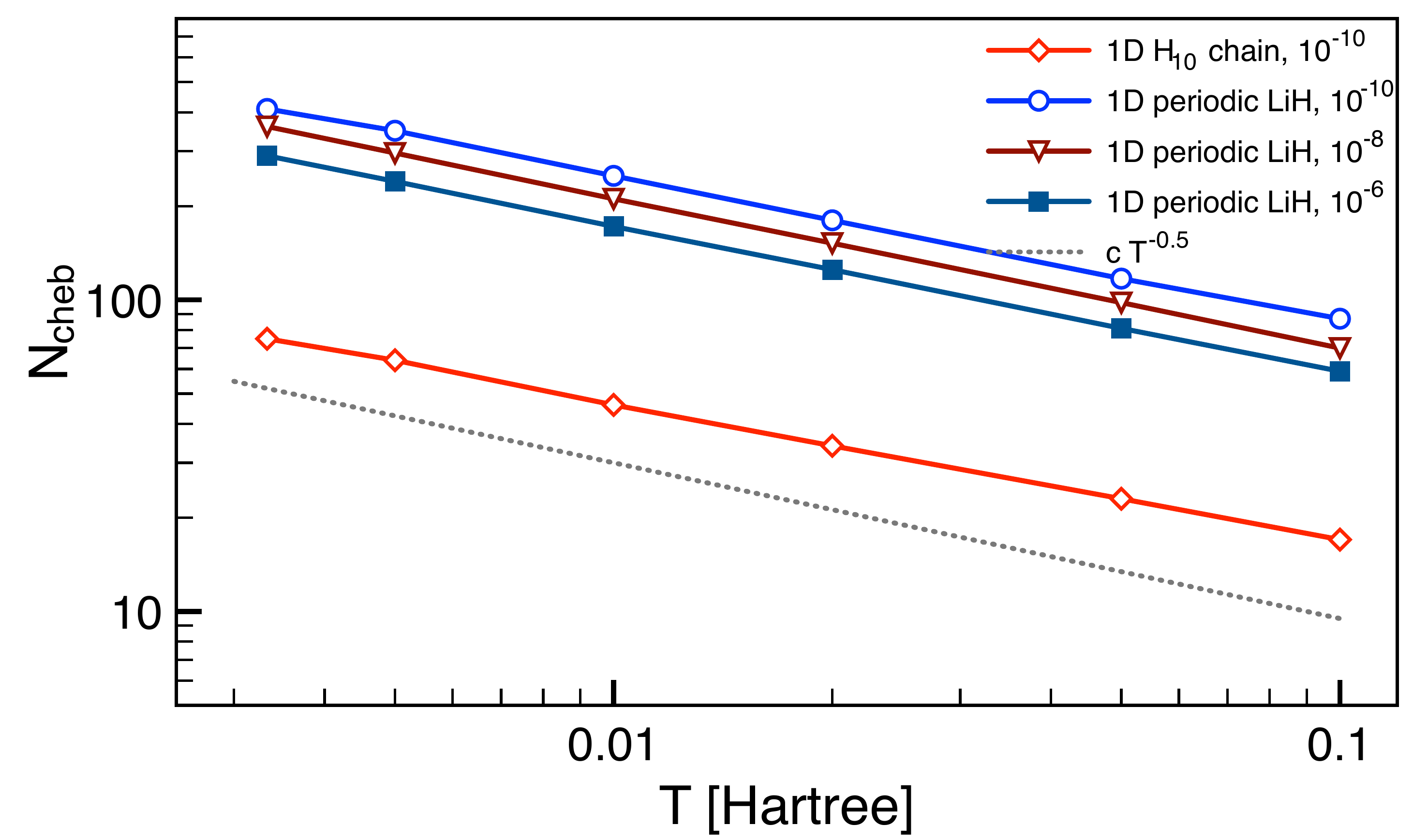}
\caption{Number of Chebyshev coefficients required to resolve the Hartree-Fock Green's function (top panel) and  the bare second order
self-energy (bottom panel) at temperature $T$ measured in Hartree up to the precision indicated, for one-dimensional H$_{10}$  and
one-dimensional periodic LiH. Parameters as specified in Table~\ref{tab:Parameters}.
}\label{fig:ConvWithJT}
\end{figure}

Fig.~\ref{fig:ConvWithJT} shows the number of Chebyshev coefficients required to reach a predetermined precision as temperature is varied.
The top panel analyzes the Hartree-Fock Green's function as the temperature is lowered, the bottom panel analyzes the second order
self-energy. The systems used are a linear chain of ten hydrogen atoms in the STO-3g basis and a periodic one-dimensional arrangement of LiH, both
systems with parameters chosen as in Table~\ref{tab:Parameters}. For the systems illustrated here, the log-log axes suggest a power law
behavior and a square-root fit shows that the scaling as a function of temperature grows slower than $\sim T^{-1/2}$, similar to observations in the context of model calculations in a Legendre basis.\cite{2018arXiv180307257C} 

\begin{figure}[tbh]
\includegraphics[width=\columnwidth]{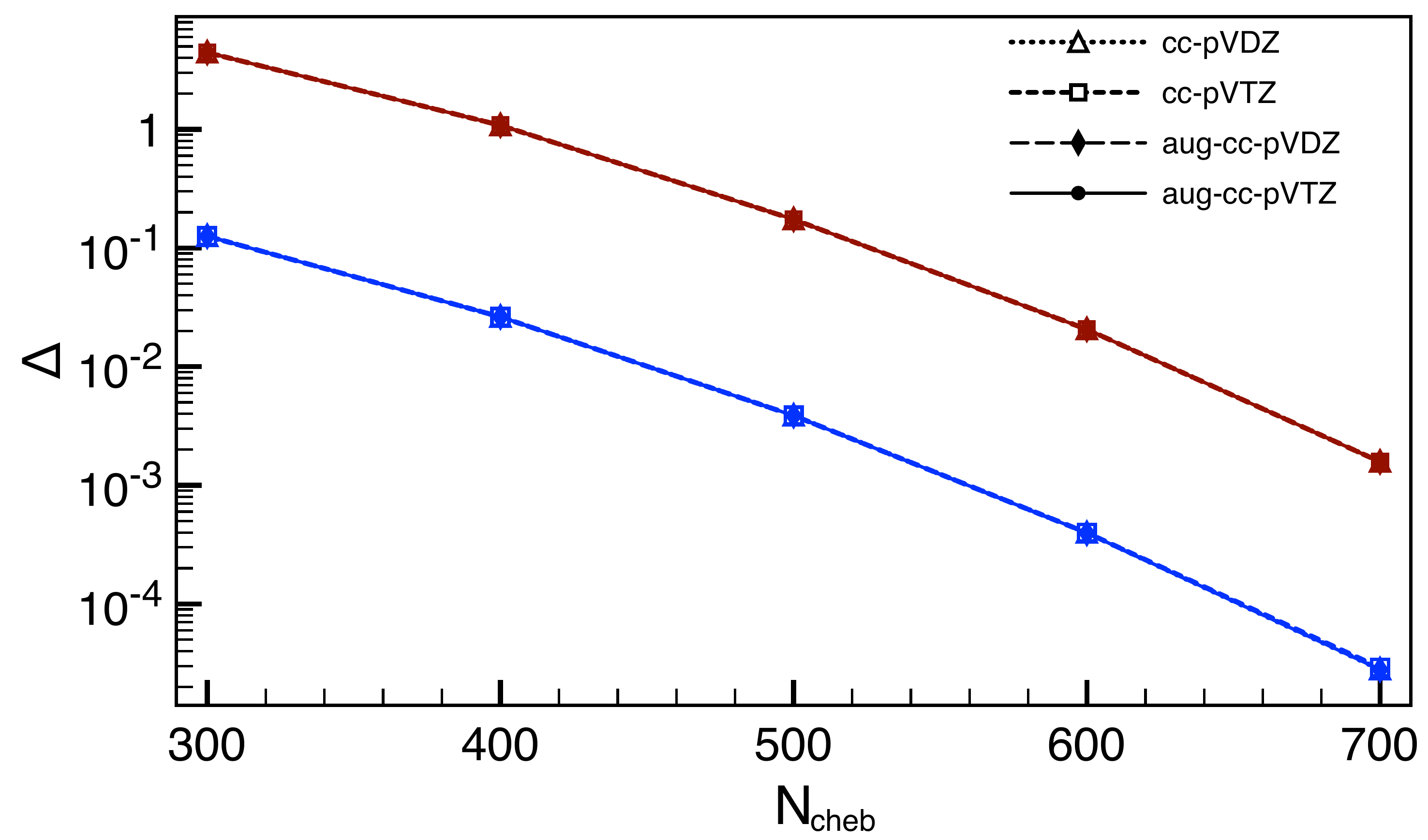}
\includegraphics[width=\columnwidth]{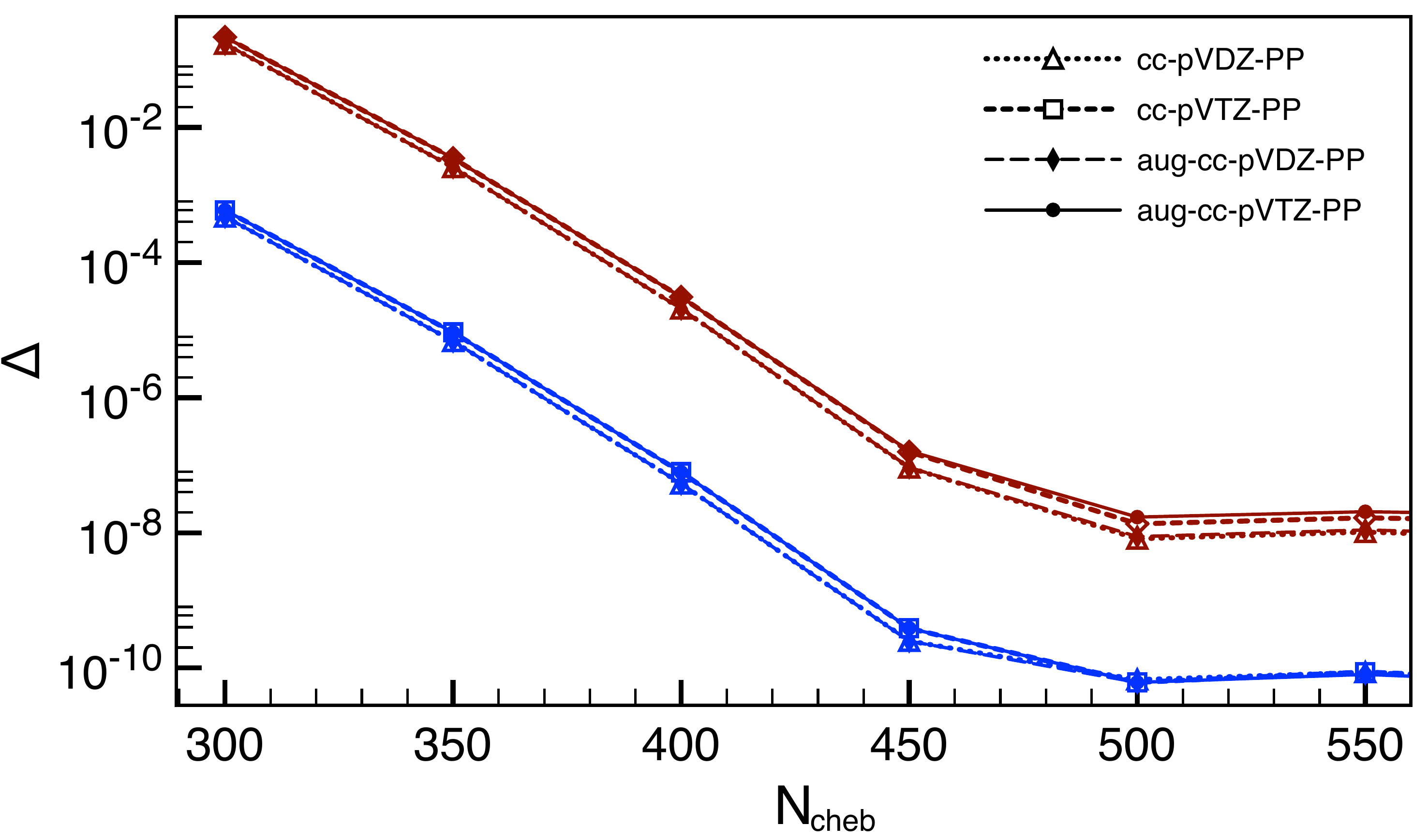}
\caption{Total (red) and maximum (blue) difference between the Chebyshev Green's function and the exact Hartree-Fock
Green's function as a function of the number of Chebyshev coefficients, for a Krypton atom in four different basis sets\cite{doi:10.1063/1.1622924,doi:10.1063/1.1622923,SchuchardtBSE} without (top panel)
and with (bottom panel) effective core potentials.}\label{fig:BasisDependence}
\end{figure}

The number of Chebyshev points required is strongly system dependent, and depends in particular on the energy spread of the atomic orbitals. 
This is illustrated in the top panel of Fig.~\ref{fig:BasisDependence}, which shows the maximum and total difference between an exactly evaluated Hartree-Fock Green's function and its Chebyshev approximant for a Kr atom as a function of the number of coefficients.
One can see that, for the bases chosen, the maximum error remains $\sim 10^{-4}$ even when $700$ components are chosen, independent of the basis.
This slow convergence is due to low-lying core states, which are fully occupied and in $\tau$-domain represented as an exponential decay
towards zero with a large decay constant.

Alternatively, choosing effective core potentials (lower panel of Fig.~\ref{fig:BasisDependence}) eliminates these low-lying states and causes a
more rapid convergence of the polynomial expansion with the number of coefficients (maximum difference of $10^{-9}$ at $\sim 450$ coefficients). Whether a
different treatment of the core states, for example by analytically modeling them with a delta function in real frequency, is more effective than a brute force expansion into Chebyshev polynomials is an open question for future
research.

\begin{figure}[tbh]
\includegraphics[width=\columnwidth]{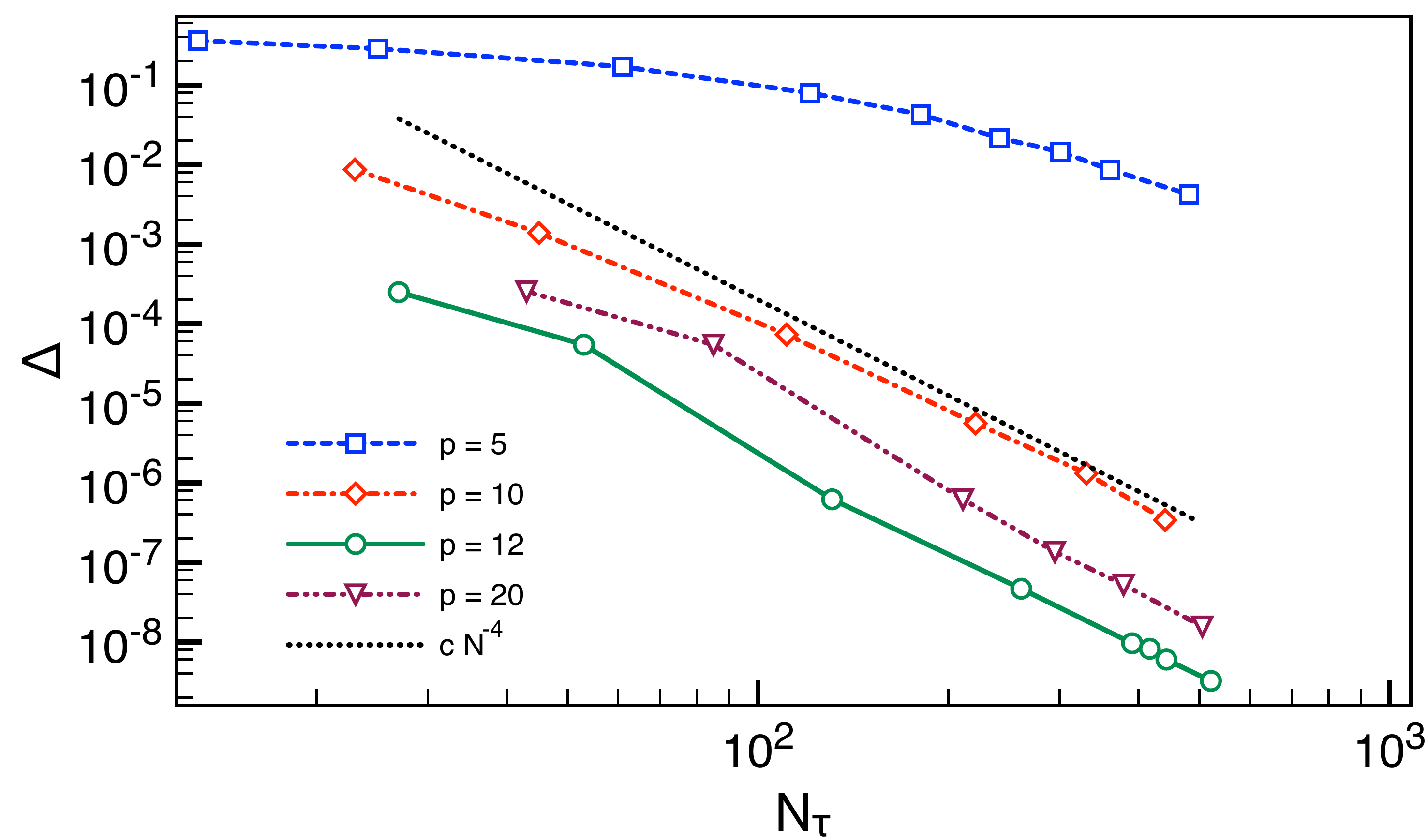}
\caption{Convergence of the power grid interpolation of a Hartree-Fock Green's function with the total number of points in the grid for the
periodic one-dimensional LiH solid. $p$ denotes the power discretization of the grid,
parameters as specified in Table~\ref{tab:Parameters}.}\label{fig:PowerGridDiff}
\end{figure}
In order to contrast these results to the commonly used uniform power grids,  Fig.~\ref{fig:PowerGridDiff} shows the convergence of power grid data
to the exact result for the periodic one-dimensional LiH solid of Fig.~\ref{fig:ChebDiff} (lower panel). Data on the power grid is interpolated using cubic splines. 
This leads to a convergence $\sim u^{-4}$ as a function of the  uniform spacing $u$. It is evident that there is an `optimal'
number of power points ($12$, in this case) which minimizes the prefactor of the convergence to the exact result (but does not change the
scaling). For results accurate to $10^{-8}$, $350$ $\tau$-points were necessary for the optimal choice of power grid parameters, around
twice as many as for the Chebyshev grid.

\begin{figure}[tbh]
\includegraphics[width=\columnwidth]{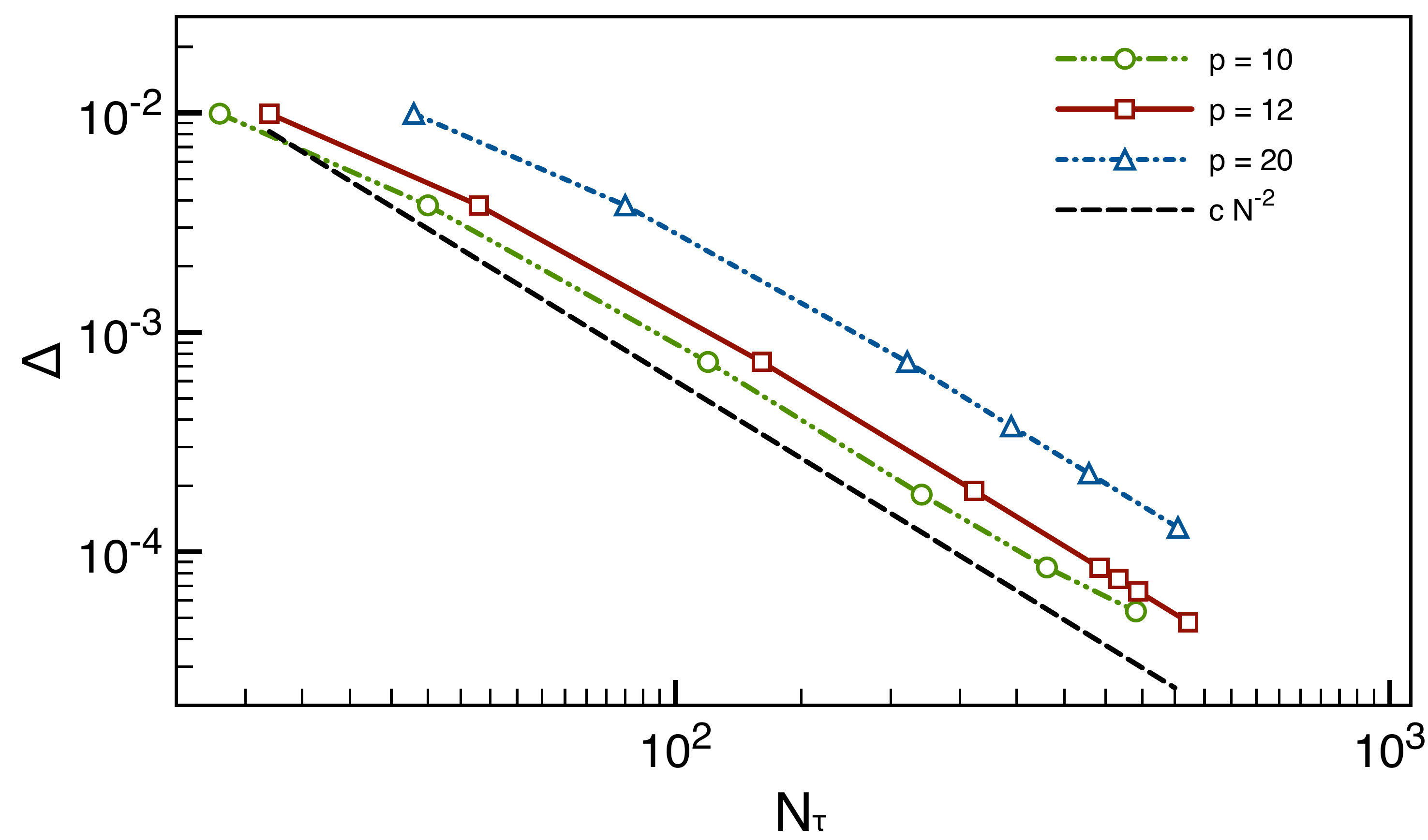}
\caption{Convergence of the Dyson equation solution for the H$_2$ molecule in discretized imaginary time using the method proposed in
Ref.~\onlinecite{Stan09}. $p$ denotes the number of power points of the grid and parameters are as
specified in Table~\ref{tab:Parameters}.}
\label{fig:DiscreteDysonDiff}
\end{figure}

The full power of the Chebyshev representation becomes apparent when Fourier integrals or a solution of the Dyson equation have to be computed for
data known in the imaginary time domain. Fig.~\ref{fig:DiscreteDysonDiff} shows the convergence of the solution of the Dyson equation using
trapezoidal integration in imaginary time, as originally proposed in Ref.~\onlinecite{Stan09}. The system studied is H$_2$ in the STO-3g basis, discrete imaginary time points are defined on a power grid with different numbers of power
points. The precision of this method is limited by the convergence rate of the trapezoidal integration of the uniform part, which is only
quadratic, such that even with $500$ time slices, only a precision of roughly $10^{-4}$ can be achieved.
In contrast, Fig.~\ref{fig:ChebCollocationDiff} shows that with the method introduced in this paper, convergence is
faster than exponential. Using around $50$ slices, a precision close to $10^{-10}$ can be reached.

Similar behavior is obtained whenever Fourier integrals need to be computed from a uniform-power grid (not shown here). Data is usually
first interpolated by a spline onto a uniform grid and then Fourier transformed to Matsubara frequencies using a fast Fourier transform.
The convergence of the spline to the exact result is the leading contribution to the error of the transform and leads to inaccuracies in the
intermediate-to-high frequency region. In contrast, the closed form of Eq.~\ref{eq:FTMatMul} avoids this interpolation step entirely.

Finally, we summarize the different aspects of basis functions for imaginary 
time Green's functions in Tab.~\ref{tab:aspects}. The comparison is subjective 
by nature, and the suitability of a given basis set will very much depend on 
the application. We list the compactness (or suitability for large realistic 
systems), the cost of constructing the basis set (cheap or expensive), the ways 
of evaluating arbitrary imaginary-time points (via interpolation, recursion, 
analytic continuation formula, fast Fourier transform, or non-equidistant FFT), 
the ways of evaluating Matsubara points, and the preferred (or so far tested) 
ways of solving the Dyson equation. 

\begin{table}[tbh]
\begin{center}
    \begin{tabular}{ | l || l | l | l | l | l |}
    \hline \hline
    Basis & Comp. & Const. & Imag & Mat & Dyson \\ \hline\hline\hline
    Uniform & No & cheap & interp. & FFT & Fourier\\ \hline
    Power & No & cheap & interp. & Fourier & Fourier\\ \hline\hline
    Chebyshev & Yes & cheap & Recursion & Sec.~\ref{sec:Fourier} & Sec.~\ref{sec:Dyson}\\ \hline
    Legendre & $\sim$Yes & cheap & Recursion & Ref.~\onlinecite{Boehnke11} & Fourier\\ \hline\hline
    IR \onlinecite{Shinaoka17Basis} & Very& exp. & AC & AC & AC\\ \hline\hline
    Matsubara & No& cheap & FFT & - & diag.\\ \hline
    Spline \onlinecite{doi:10.1021/acs.jctc.5b00884}& $\sim$Yes& exp. & NFFT & interp. &diag.\\ \hline
    \hline
    \end{tabular}
\end{center}
\caption{\label{tab:aspects}Subjective comparison of different aspects of the 
various basis sets for finite-T Green's functions. The basis sets vastly differ 
in their Compactness (Comp.), basis construction effort (Const.), way of 
evaluating imaginary time (Imag.) or Matsubara (Mat.) Green's function values, 
and ways of solving the Dyson equation (via Fourier to Matsubara space, where the equation is diagonal, or as described in the main text).}
\end{table}

\begin{figure}[tbh]
\includegraphics[width=\columnwidth]{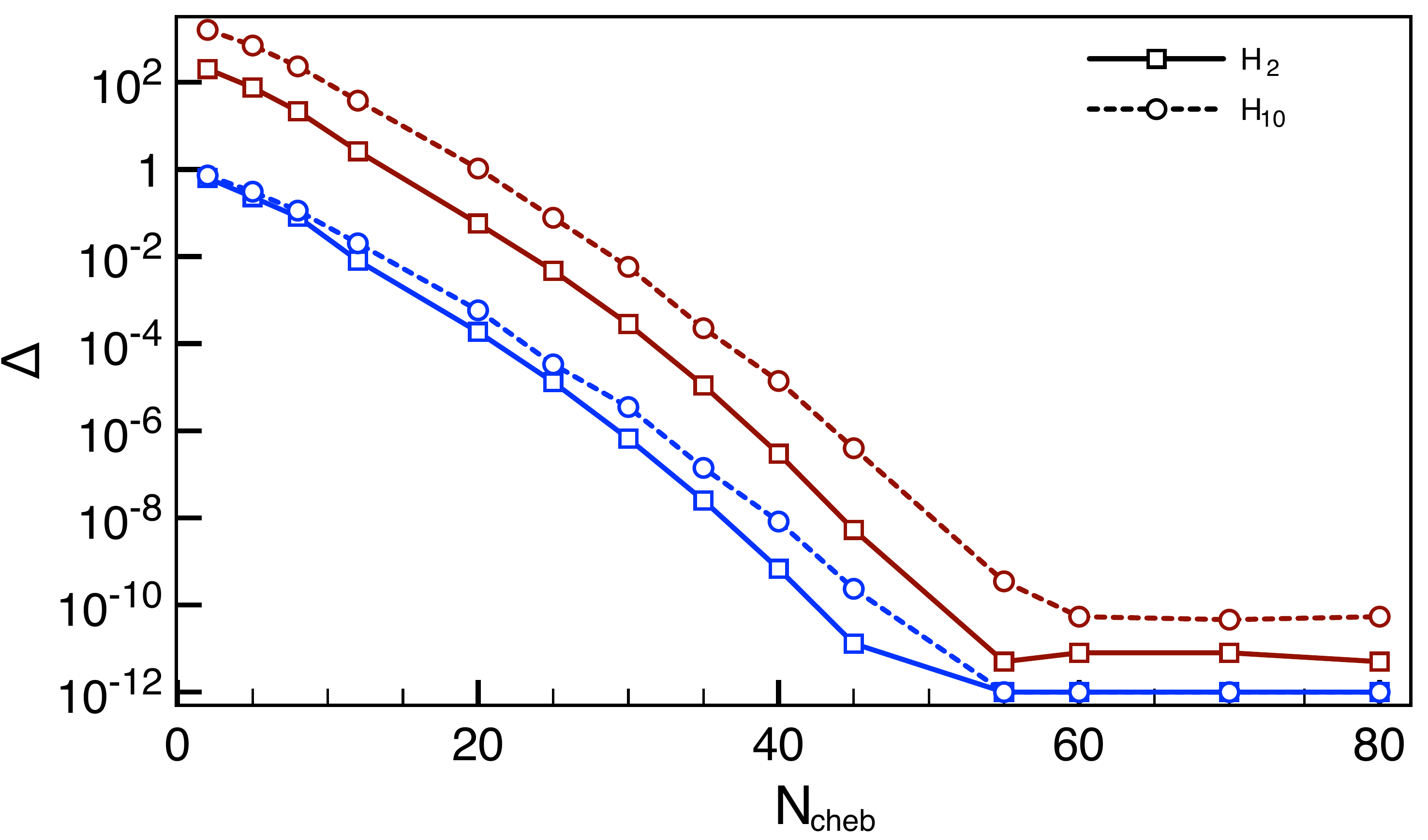}
\caption{Convergence of the Dyson equation solution for H$_2$ and H$_{10}$ with the number of Chebyshev polynomials with parameters specified in Table~\ref{tab:Parameters}.
Red lines: Sum of errors. Blue lines: Maximum error. Squares: H$_2$. Circles: H$_{10}$.}\label{fig:ChebCollocationDiff}
\end{figure}

\section{Conclusions}\label{sec:conclusions}
In conclusion, we have explored the use of an orthogonal polynomial basis for imaginary time Green's functions in the context of
realistic materials. We have observed the exponential convergence guaranteed by the analytic nature of Green's functions in practice, and
shown that for typical systems, substantially fewer imaginary time points are needed than for a uniform-power grid. The convergence rate of
the expansion depends on the system. While low-lying core states present a difficulty for this basis and lead to a slow convergence of the
expansion, the complex spectral behavior near the chemical potential is well captured by the first $50$-$100$ coefficients in the systems examined here.
We have also shown that convolutions, Dyson equations, and Fourier transforms, which correspond to commonly used operations on imaginary
time Green's functions, can be performed accurately and efficiently.
This paves the way for using Chebyshev approximated imaginary time Green's functions in calculations of realistic and model systems,
replacing the uniform and uniform-power grids that are so far being used in this context.

\acknowledgments{
EG and SI were supported by the Simons Foundation via the Simons Collaboration on the Many-Electron Problem. 
IK was supported by DOE ER 46932, DZ and AAR by NSF-CHE-1453894.
This research used resources
of the National Energy Research Scientific Computing Center, a DOE Office of Science User Facility supported by the Office of Science of the
U.S. Department of Energy under Contract No. DE-AC02-05CH11231.
}
\bibliographystyle{apsrev4-1}
\bibliography{refs}

\appendix
\section{Convolution of Chebyshev polynomials and $t$-coefficients}
Chebyshev interpolation can be used to solve Fredholm integral equations with piecewise-continuous convolution kernels. Solution of such
problems requires knowledge of a special system of coefficients $t_{kl}^j$ defined according to
\begin{equation}\label{t_coeff}
    T_{kl}(x) = \sump_{j} t_{kl}^j T_j(x).
\end{equation}
\begin{subequations}
    \begin{align}
        T_{kl}(x) &= T_{kl}^{-}(x) \pm T_{kl}^{+}(x), \label{Decoup} \\
        T_{kl}^{-}(x) &\equiv \int_{-1}^{x} T_k(x-\tau-1) T_l(\tau) d\tau, \label{Int1} \\
        T_{kl}^{+}(x) &\equiv -\int_{x}^{1} T_k(x-\tau+1) T_l(\tau) d\tau. \label{Int2}
    \end{align}
\end{subequations}

The plus (minus) sign in equation \ref{Int2} corresponds to fermionic (bosonic) Green's function case.
In the following derivations we will use the product formula for $T_m(x)$,
\begin{equation}\label{T_product}
    T_m(x) T_n(x) = \frac{1}{2} \left[T_{m+n}(x) + T_{|m-n|}(x)\right],
\end{equation}
and, in particular, its special case of $n=1$,
\begin{equation}\label{T_rec}
    T_{m+1}(x) = 2x T_{m}(x) - T_{|m-1|}(x).
\end{equation}

The Chebyshev polynomials $T_m(x)$ form an orthogonal system of functions on segment $[-1;1]$ w.r.t. scalar product
\begin{align}
    &\left\langle f(x), g(x) \right\rangle =
    \int_{-1}^1 f(x) g(x) \frac{dx}{\sqrt{1-x^2}},\\
&\left\langle T_m(x), T_n(x) \right\rangle =
\pi\frac{1+\delta_{n,0}}{2}\delta_{m,n}.
\end{align}
The orthogonality property means that two equal Chebyshev polynomial expansions have to be equal order by order.

Another result we will need is the indefinite integral of the Chebyshev polynomials,
\begin{equation}
\int T_n(x) dx =
\left\{
\begin{array}{ll}
\frac{1}{2}\left(\frac{T_{n+1}(x)}{n+1} + \frac{T_{n-1}(x)}{n-1}\right) + C, &n\neq1, \\
\frac{T_2(x) + T_0(x)}{4} + C, &n=1.
\end{array}\right.
\end{equation}

\subsection{Symmetry properties of $T_{kl}(x)$}

\begin{theorem}\label{theorem_kl_sym}
Functions $T_{kl}(x)$ are symmetric w.r.t. their indices, $T_{lk}(x) = T_{kl}(x)$.
\end{theorem}
\begin{proof}\renewcommand{\qedsymbol}{}
\begin{align*}
    T^-_{lk}(x) &= \int_{-1}^{x} T_l(x-\tau-1) T_k(\tau) d\tau =
     \left\{\substack{x-\tau-1 = t\\dt = -d\tau}\right\}\\ &=
     -\int_{x}^{-1} T_l(t) T_k(x-t-1)dt \\
     &= \int_{-1}^{x} T_k(x-t-1) T_l(t)dt = T^-_{kl}(x).
\end{align*}
Similarly, $T^+_{lk}(x) = T^+_{kl}(x)$, and, according to definition (\ref{Decoup}), $T_{lk}(x) = T_{kl}(x)$.
\end{proof}
\begin{corollary}
    $t$-coefficients are symmetric w.r.t. their lower indices,
    $t^j_{lk} = t^j_{kl}$.
\end{corollary}

\begin{theorem}\label{theorem_parity}
Functions $T_{kl}(x)$ are either even or odd functions, depending on the value 
of $k+l$, $T_{kl}(-x) = (-1)^{k+l+1} T_{kl}(x)$.
\end{theorem}
\begin{proof}\renewcommand{\qedsymbol}{}
Using symmetry property of Chebyshev polynomials, $T_k(-x) = (-1)^k T_k(x)$,
we can write
\begin{align*}
T_{kl}^-(-x) &= \int_{-1}^{-x}T_k(-x-\tau-1)T_l(\tau)d\tau \\&=
              (-1)^k\int_{-1}^{-x}T_k(x+\tau+1)T_l(\tau)d\tau \\&=
              \left\{\substack{-\tau = t\\dt = -d\tau}\right\} =\\
             &= -(-1)^k\int_{1}^{x}T_k(x-t+1)T_l(-t)dt\\&
              = (-1)^{k+l+1}\int_{1}^{x}T_k(x-t+1)T_l(t)dt \\
             &= -(-1)^{k+l+1}\int_{x}^{1}T_k(x-t+1)T_l(t)dt\\&
              = (-1)^{k+l+1} T^+_{kl}(x).
\end{align*}
Similarly, $T_{kl}^+(-x) = (-1)^{k+l+1}T_{kl}^-(x)$.
\begin{align*}
T_{kl}(-x) &= T_{kl}^-(-x) \pm T_{kl}^+(-x) \\&=
(-1)^{k+l+1}(T_{kl}^+(x) \pm T_{kl}^-(x)) = (-1)^{k+l+1}T_{kl}(x).
\end{align*}
\end{proof}
\begin{corollary}
    $t^j_{kl} = 0$ for even values of $k + l + j$.
\end{corollary}
\begin{proof}\renewcommand{\qedsymbol}{}
Chebyshev polynomials $T_k(x)$ are even functions for even $k$ and odd 
functions for odd $k$. Therefore, expansion (\ref{t_coeff}) can contain only even $j$ terms when $T_{kl}(x)$ is even and only odd $j$ terms
when $T_{kl}(x)$ is odd. Combining this observation with the proven theorem, we conclude that terms with different parities of $k + l + 1$
and $j$ do not contribute to the sum (\ref{t_coeff}).
\end{proof}

\subsection{Recurrence relation for the convolutions $T^{\pm}_{kl}(x)$}

Chebyshev polynomials of shifted argument fulfill the following recursion relation,
\begin{align*}
    T_{k+1}(x-\tau\pm1) &= 2x T_{k}(x-\tau\pm1) - 2\tau T_{k+1}(x-\tau\pm1) \\&
    \pm2 T_{k+1}(x-\tau\pm1) - T_{|k-1|}(x-\tau\pm1). 
\end{align*}

Therefore, integrands in Eqs.~(\ref{Int1},\ref{Int2}) can be expressed as
\begin{widetext}
\begin{align}
  T_{k+1}(x-\tau\pm 1) T_l(\tau) = 2xT_{k}(x-\tau\pm 1) T_l(\tau) - &T_k(x-\tau\pm 1)\left[2\tau T_{l}(\tau)\right] \nonumber \\
  & \pm 2T_k(x-\tau\pm 1)T_l(\tau) - T_{|k-1|}(x-\tau\pm 1)T_l(\tau) \nonumber \\
  = 2xT_{k}(x-\tau\pm 1) T_l(\tau) - &T_k(x-\tau\pm 1)\left[T_{l+1}(\tau) + T_{|l-1|}(\tau)\right] \nonumber \\
  & \pm 2T_k(x-\tau\pm 1)T_l(\tau) - T_{|k-1|}(x-\tau\pm 1)T_l(\tau).
\end{align}
\end{widetext}
Plugging this into Eq.~\eqref{Int1} or \eqref{Int2} we get a recurrence relation for
$T_{k,l}^{\pm}(x)$,
\begin{align}
  T_{k+1,l}^{\pm}(x) &= 2x T_{k,l}^{\pm}(x) - T_{k,l+1}^{\pm}(x) \nonumber\\&- T_{k,|l-1|}^{\pm}(x) \pm 2 T_{k,l}^{\pm}(x) -
  T_{|k-1|,l}^{\pm}(x),
  \label{Rec}
\end{align}
with boundary conditions
\begin{align*}
  T_{k,0}^{-}(x) &= T_{0,k}^{-}(x) = \int_{-1}^{x} T_k(x-\tau -1)d\tau \\&= \left\{\substack{x-\tau - 1 = t\\dt = -d\tau}\right\} =
  -\int_{x}^{-1} T_k(t)dt = \int_{-1}^{x} T_k(t) dt,\\
  T_{k,0}^{+}(x) &= T_{0,k}^{+}(x) = -\int_{x}^{1} T_k(x-\tau +1)d\tau \\&= \left\{\substack{x-\tau + 1 = t\\dt = -d\tau}\right\} =
  \int_{1}^{x} T_k(t)dt = -\int_{x}^{1} T_k(t) dt,
\end{align*}
that lead to
\begin{align}
    T_{k,0}^{-}(x) &= T_{0,k}^{-}(x) = \nonumber \\
    &\begin{cases}
        \frac{T_2(x) - T_0(x)}{4}, &k=1 \\
        \frac{1}{2}\left(\frac{T_{k+1}(x)+(-1)^k}{k+1} - \frac{T_{|k-1|}(x)+(-1)^k}{k-1}\right), &k\neq1
    \end{cases}
\end{align}
\begin{align}
  T_{k,0}^{+}(x) &= T_{0,k}^{+}(x) = \nonumber \\
  &\begin{cases}
      \frac{T_2(x) - T_0(x)}{4}, &k=1 \\
      \frac{1}{2}\left(\frac{T_{k+1}(x)-1}{k+1} - \frac{T_{|k-1|}(x)-1}{k-1}\right), &k\neq 1
  \end{cases}
\end{align}
or generally,
\begin{align}
  T_{k,0}^{\pm}(x) &= T_{0,k}^{\pm}(x) = \nonumber \\
  &\begin{cases}
    \frac{T_2(x) - T_0(x)}{4}, &k=1 \\
    \frac{1}{2}\left(\frac{T_{k+1}(x)-(\pm 1)^{k+1}}{k+1} - \frac{T_{|k-1|}(x)-(\pm 1)^{k-1}}{k-1}\right), &k\neq 1.
  \end{cases}
\end{align}

In order to apply recurrence~\ref{Rec}, one needs expressions for $T_{1,l}(x)$
in addition to the boundary values given above. These can be derived from 
Eq.~\ref{Rec} with $k=0$,
\begin{align}
  T_{1,l}(x) = \frac{1}{2} \left[2x T^{\pm}_{0,l}(x) - T_{0,l+1}^{\pm}(x) - T_{0,|l-1|}^{\pm}(x)  \pm T^{\pm}_{0,l}(x) \right].
\end{align}

The complete set of equations reads as follows,
\begin{widetext}
\begin{align}
  T_{k,0}^{\pm}(x) &= T_{0,k}^{\pm}(x) =
  \frac{1}{2}\left(\frac{T_{k+1}(x)-(\pm 1)^{k+1}T_0(x)}{k+1} -
  \frac{T_{|k-1|}(x)-(\pm 1)^{k-1}T_0(x)}{k-1}(1-\delta_{k,1})\right),\nonumber\\
  T_{1,l}(x) &= \frac{1}{2} \left[2x T^{\pm}_{0,l}(x) - T_{0,l+1}^{\pm}(x) - T_{0,|l-1|}^{\pm}(x)  \pm T^{\pm}_{0,l}(x) \right] \label{ChebFullRel},\\
  T_{k+1,l}^{\pm}(x) &= 2x T_{k,l}^{\pm}(x) - T_{k,l+1}^{\pm}(x) - T_{k,|l-1|}^{\pm}(x) \pm 2 T_{k,l}^{\pm}(x) - T_{|k-1|,l}^{\pm}(x).\nonumber
\end{align}
\end{widetext}

\subsection{Chebyshev expansion of the convolution.}

In this paragraph we derive a recurrence relation that allows for
efficient evaluation of $t_{k,l}^j$ in constant time per coefficient.
Let us consider Chebyshev expansion of Eq.~\ref{ChebFullRel},
\begin{align}
  T_{k,l}^{\pm}(x) &= \sump_{j=0}^{l+k+1} t_{k,l}^{j(\pm)}T_{j}(x).
\end{align}
Summation limit $k+l+1$ comes from the fact that $T_k(x-\tau\pm1)T_l(\tau)$
is a polynomial of degree at most $k+l$ in $\tau$, and the integration in
(\ref{Int1},\ref{Int2}) adds one to the degree.

First, consider the boundary conditions. In case of $l=0$ the coefficients $t_{k,0}^{j(\pm)}$ are obtained from the solution of the following
equations,
\begin{align*}
   \sump_{j=0}^{k+1} t_{k,0}^{j(\pm)}T_{j}(x) &=
   \frac{T_{k+1}(x)-(\pm 1)^{k+1}T_0(x)}{2(k+1)} - \\
   &\frac{T_{|k-1|}(x)-(\pm 1)^{k-1}T_0(x)}{2(k-1)}(1-\delta_{k,1}).
\end{align*}
$k = 1$:
\begin{align*}
  t_{1,0}^{j(\pm)} =
  \begin{cases}-\frac{1}{2}, &j=0\\
  \frac{1}{4}, &j =2 \\
  0, &\mathrm{otherwise}.
\end{cases}
\end{align*}
$k \neq 1$:
\begin{align*}
  \sump_{j=0}^{k+1} t_{k,0}^{j(\pm)}&T_{j}(x) = \frac{T_{k+1}(x)-(\pm 1)^{k+1}T_0(x)}{2(k+1)}  \\
     &- \frac{T_{|k-1|}(x)-(\pm 1)^{k-1}T_0(x)}{2(k-1)}\\
  	 &= \frac{T_{k+1}(x)}{2k + 2} - \frac{T_{|k-1|}(x)}{2k - 2}
     + \frac{2(\pm 1)^{k+1}}{k^2 - 1}\frac{T_0(x)}{2}.
\end{align*}
By grouping Chebyshev polynomials of the same order we get the following expressions,
\begin{align*}
  t_{k,0}^{j(\pm)} = \begin{cases}
                       \frac{2(\pm 1)^{k+1}}{k^2-1}, &j = 0 \\
                       -\frac{1}{2k-2}, &j = k -1, k>0 \\
                       \frac{1}{2k+2}, &j = k + 1, k>0 \\
                       1, &k = 0, j = 1.
  \end{cases}
\end{align*}
$l>0, k=1$:
\begin{align*}
  \sump_{j=0}^{l+2}&t_{1,l}^{j(\pm)} T_j(x) =
  \frac{1}{2}\sump_{j=0}^{l+2} \bigg[
  \frac{t_{0,l}^{j-1(\pm)}}{2^{\delta_{j,1}}}\delta_{j>0} \\&+
  \frac{1}{2}t_{0,l}^{0(\pm)}\delta_{j,1} +
  2^{\delta_{j,0}}t_{0,l}^{j+1(\pm)}\delta_{j\leq l}
  \\&- t_{0,l+1}^{j(\pm)}
  - t_{0,l-1}^{j(\pm)} \delta_{j\leq l}
  \pm 2 t_{0,l}^{j(\pm)} \delta_{j\leq l+1} \bigg]T_{j}(x).
\end{align*}
General case $l>0, k>0$:
\begin{align*}
  \sump_{j=0}^{k+l+2}&t_{k+1,l}^{j(\pm)} T_j(x) = \sump_{j=0}^{k+l+2} \bigg[
    \frac{t_{k,l}^{j-1(\pm)}}{2^{\delta_{j,1}}}\delta_{j>0} \\&+
    \frac{1}{2}t_{k,l}^{0(\pm)}\delta_{j,1} +
    2^{\delta_{j,0}}t_{k,l}^{j+1(\pm)}\delta_{j\leq k+l} \\ &
    - t_{k,l+1}^{j(\pm)}
    - t_{k,l-1}^{j(\pm)} \delta_{j\leq k+l}
    \pm 2 t_{k,l}^{j(\pm)} \delta_{j\leq k+l+1}\\ &
    - t_{k-1,l}^{j(\pm)} \delta_{j\leq k+l}
\bigg] T_{j}(x).
\end{align*}

Now we can summarize the results for the $t^{j(\pm)}_{kl}$ coefficients:
\begin{subequations}
\begin{align}
    k=1&,l=0:\nonumber\\&t_{1,0}^{j(\pm)} =
    \begin{cases}
        -\frac{1}{2}, &j = 0 \\
        \frac{1}{4}, &j = 2 \\
        0, &\mathrm{otherwise}.
    \end{cases}\\
    k\neq 1&,l=0:\nonumber\\&t_{k,0}^{j(\pm)} =
    \begin{cases}
        \frac{2(\pm 1)^{k+1}}{k^2-1}, &j = 0 \\
        -\frac{1}{2k-2}, &j = k -1, k>0 \\
        \frac{1}{2k+2}, &j = k + 1, k>0 \\
        1, &k = 0, j = 1.
    \end{cases}\\
    k=1&,l>0:\nonumber\\&t_{1,l}^{j(\pm)} = \frac{1}{2}
    \frac{t_{0,l}^{j-1(\pm)}\delta_{j>0}}{2^{\delta_{j,1}}} + 
    \frac{t_{0,l}^{0(\pm)}\delta_{j,1}}{2}
    +2^{\delta_{j,0}}t_{0,l}^{j+1(\pm)}\delta_{j\leq l}\nonumber \\ 
    &- t_{0,l+1}^{j(\pm)}
    - t_{0,l-1}^{j(\pm)} \delta_{j\leq l}
    \pm 2 t_{0,l}^{j(\pm)} \delta_{j\leq l+1}. \\
    k>0&,l>0:\nonumber\\&t_{k+1,l}^{j(\pm)} =
    \frac{t_{k,l}^{j-1(\pm)}\delta_{j>0}}{2^{\delta_{j,1}}} +
    \frac{t_{k,l}^{0(\pm)}\delta_{j,1}}{2} +
    2^{\delta_{j,0}}t_{k,l}^{j+1(\pm)}\delta_{j\leq k+l}\nonumber\\
    &- t_{k,l+1}^{j(\pm)}
    - t_{k,l-1}^{j(\pm)} \delta_{j\leq k+l}
    \pm 2 t_{k,l}^{j(\pm)} \delta_{j\leq k+l+1}\nonumber\\&
    - t_{k-1,l}^{j(\pm)} \delta_{j\leq k+l}.
\end{align}
\end{subequations}

And the final expression for $t$-coefficients will be
\begin{align}
t^{j}_{kl} = t^{j(-)}_{kl} \pm t^{j(+)}_{kl},
\end{align}
where the plus (minus) sign corresponds to the fermionic (bosonic) convolution.

\section{Comparison to the Legendre basis}
Ref.~\onlinecite{Boehnke11} proposed the use of Legendre polynomials as a basis 
for Monte Carlo Green's functions. These Green's functions could in principle 
be adapted to the real materials context. Fig.~\ref{fig:CompLegendre} shows the 
maximum error curves of Fig.~\ref{fig:ChebDiff} in the main text. It is evident 
that the Legendre basis also converges exponentially, but is substantially less 
compact than the Chebyshev basis. This is expected from the minmax properties 
of the Chebyshev polynomials. It is also the expected behavior for core states, 
as the zeros of the Chebyshev polynomials are closer to $-1$ and $1$ than those 
of Legendre polynomials of the same order,\cite{Szego36} and are therefore better able to 
resolve the exponential decay core states and high lying excitations.

The two types of polynomials differ in several technical aspects. For instance, 
the roots of Chebyshev polynomials are known analytically, whereas 
those of Legendre polynomials need to be evaluated numerically and tabulated. 
On the other hand, the orthogonality relation of the Legendre polynomials is 
easier.

\begin{figure}[tbh]
\includegraphics[width=\columnwidth]{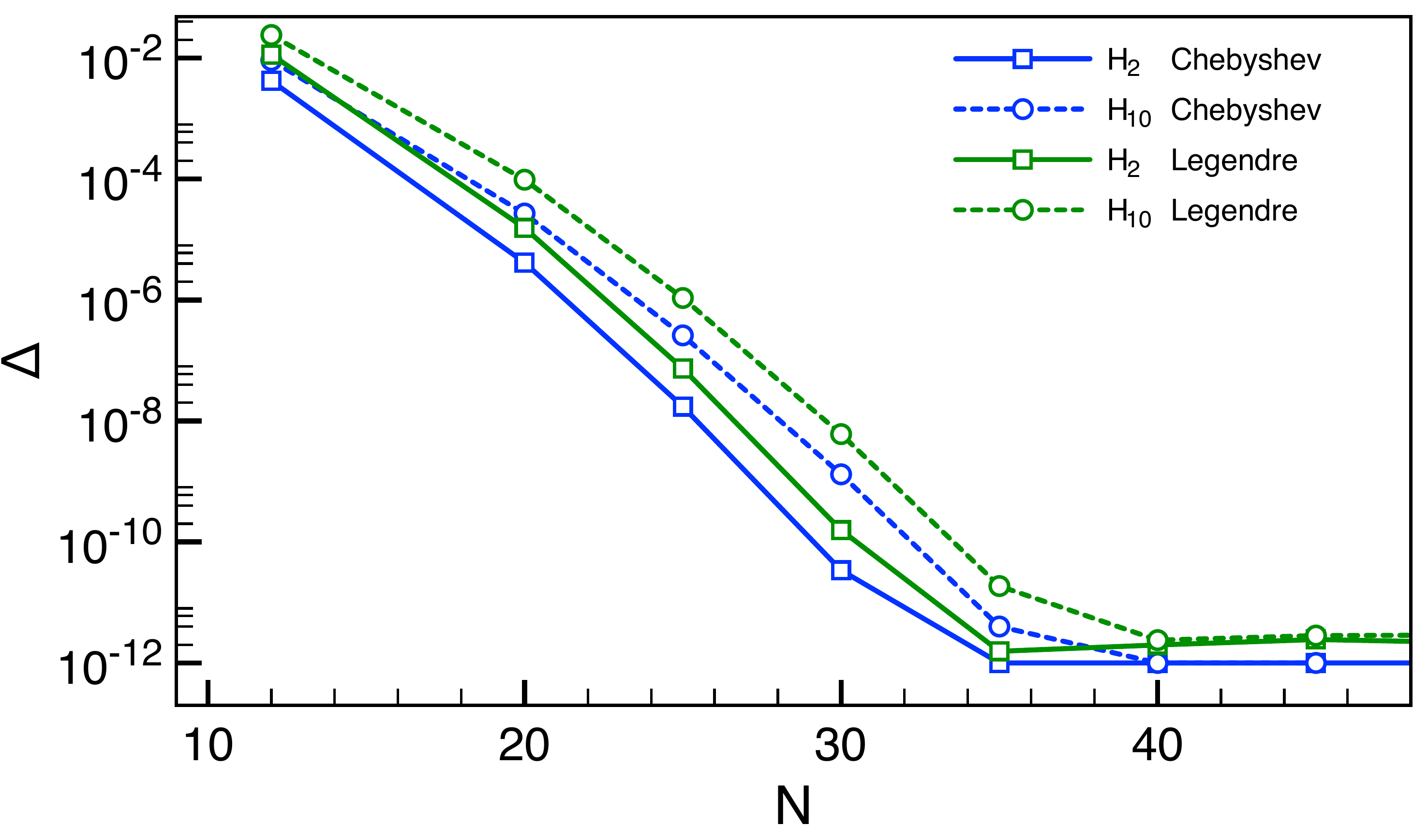}
\caption{Convergence of the Hartree-Fock Green's function with the number of 
Chebyshev and Legendre expansion coefficients. Blue curves correspond to the 
maximum difference of the Chebyshev expansion when compared to the exact result 
(Fig.~\ref{fig:ChebDiff}). Green curves show the same data, for the Legendre 
basis.\cite{Boehnke11} H$_2$ molecule (open square) and H$_{10}$ chain (open circle) parameters 
as specified in Table~\ref{tab:Parameters}.}
\label{fig:CompLegendre}
\end{figure}
\end{document}